\documentclass[submission,copyright,creativecommons]{eptcs}
\providecommand{\event}{GandALF 2016}

\usepackage[utf8]{inputenc}
\usepackage{cmap}
\usepackage[T1]{fontenc}
\usepackage{lmodern}
\usepackage[english]{babel}
\usepackage[autostyle=true]{csquotes}
\MakeOuterQuote{"}

\usepackage{enumerate}
\usepackage[protrusion=true,expansion=true]{microtype}
\usepackage{amsfonts,amssymb,xcolor}
\usepackage{stmaryrd}
\SetSymbolFont{stmry}{bold}{U}{stmry}{m}{n}
\usepackage[fixamsmath,disallowspaces]{mathtools}
\usepackage{fixmath}

\usepackage{fixltx2e}
\usepackage{scrhack}
\usepackage{relsize}
\usepackage{braket} 
\hypersetup{
    colorlinks=true,
    linkcolor={red!50!black},
    citecolor={blue!50!black},
    urlcolor={blue!80!black}
}

\usepackage[all]{hypcap}

\makeatletter
\g@addto@macro{\UrlBreaks}{\UrlOrds}
\makeatother

\usepackage[xspace]{ellipsis}

\usepackage[
  section    ]{placeins}

\usepackage{pdfcolmk}

\makeatletter

\let\c@theorem\relax

\let\c@lemma\relax

\let\c@corollary\relax

\let\c@definition\relax

\let\c@example\relax

\makeatother

\usepackage{amsthm}
\usepackage{aliascnt}
\usepackage{cleveref} \usepackage{booktabs}

\newcommand{\ie}{i.e.\@\xspace}
\newcommand{\eg}{e.g.\@\xspace}
\newcommand{\wrt}{w.\,r.\,t.\@\xspace}
\newcommand{\suchthat}{s.\,t.\@\xspace}

\newcommand{\wloss}{w.l.o.g.\@\xspace}
\newcommand{\Wloss}{W.l.o.g.\@\xspace}

      \newcommand{\complClFont}[1]{\mathbf{#1}}                    \newcommand{\logicClFont}[1]{\mathsf{#1}}                 \newcommand{\problemFont}[1]{\mathrm{#1}}         \newcommand{\mathCommandFont}[1]{\mathrm{#1}}     
\newcommand{\bigO}[1]{\protect\ensuremath{{\mathcal{O}\left(#1\right)}}}

\newcommand{\arity}[1]{{\protect\ensuremath{\mathCommandFont{ar}(#1)}}}
\newcommand{\rel}[1]{{\protect\ensuremath{\mathCommandFont{rel}(#1)}}}

\newcommand{\tuple}[1]{\vec{#1}}
\newcommand{\imp}{\rightarrow}
\newcommand{\Dom}[1]{{\protect\ensuremath{\mathsf{Dom}(#1)}}}
\newcommand{\Fr}[1]{{\protect\ensuremath{\mathsf{Fr}(#1)}}}
\newcommand{\degg}[1]{{\protect\ensuremath{\mathsf{deg}_\negg(#1)}}}
\newcommand{\Prop}[1]{{\protect\ensuremath{\mathsf{Prop}(#1)}}}

\newcommand {\U}{\mathsf{U}\,}

\newcommand{\negg}{{\sim}}

\providecommand{\dfn}{\mathrel{\mathop:}=}
\providecommand{\ddfn}{\mathrel{\mathop{{\mathop:}{\mathop:}}}=}
\newcommand{\fr}{\ensuremath{\mathrm{Fr}}}

\newcommand{\ADQBF}{\protect\ensuremath\problemFont{ADQBF}}
\newcommand{\DQBF}{\protect\ensuremath\problemFont{DQBF}}

\newcommand{\SAT}{\protect\ensuremath\problemFont{SAT}}
\newcommand{\TRUE}{\protect\ensuremath\problemFont{TRUE}}
\newcommand{\VAL}{\protect\ensuremath\problemFont{VAL}}

\newcommand{\leqpm}{\protect\ensuremath{\leq_\mathCommandFont{m}^\mathCommandFont{P}}}

\newcommand{\leqlogm}{\protect\ensuremath{\leq^\mathCommandFont{log}_\mathCommandFont{m}}}

\newcommand{\AEXPPOLY}{\protect\ensuremath{\complClFont{AEXPTIME}(\mathrm{poly})}\xspace}
\newcommand{\PSPACE}{\protect\ensuremath{\complClFont{PSPACE}}\xspace}
\newcommand{\NEXPTIME}{\protect\ensuremath{\complClFont{NEXPTIME}}\xspace}
\newcommand{\EXPTIME}{\protect\ensuremath{\complClFont{EXPTIME}}\xspace}
\newcommand{\SigmaE}[1]{{\ensuremath\protect\Sigma^\mathCommandFont{E}_{#1}}}
\newcommand{\PiE}[1]{{\ensuremath\protect\Pi^\mathCommandFont{E}_{#1}}}

\newcommand{\calG}{\protect\ensuremath{\mathcal{G}}}

\newcommand{\ESO}{\logicClFont{ESO}}
\newcommand{\SO}{\logicClFont{SO}}

\newcommand{\TL}{\logicClFont{TL}}
\newcommand{\PL}{\logicClFont{PL}}
\newcommand{\PD}{\logicClFont{PD}}
\newcommand{\PLInc}{\logicClFont{PLInc}}
\newcommand{\ML}{\logicClFont{ML}}
\newcommand{\MInc}{\logicClFont{MInc}}
\newcommand{\QPL}{\logicClFont{QPL}}
\newcommand{\PTL}{\logicClFont{PTL}}
\newcommand{\QPTL}{\logicClFont{QPTL}}
\newcommand{\QPD}{\logicClFont{QPD}}
\newcommand{\QPLInc}{\logicClFont{QPLInc}}

\newcommand{\LL}{\logicClFont{L}}
\newcommand{\modelsPL}{\models_{\PL}}
\newcommand{\dep}[1]{\mathrm{dep}\!\left(#1\right)}
\newcommand{\depop}{\mathrm{dep}}
\newcommand{\IF}{\logicClFont{IF}}

\newcommand{\SOQPL}{\logicClFont{SO}_2}
\newcommand{\ESOQPL}{\logicClFont{ESO}_2}
\newcommand{\ESOQPLuniq}{\logicClFont{ESO}_2^\textrm{u}}
\newcommand{\SigmaQPL}[1]{\Sigma_{#1}\text{-}\logicClFont{SO}_2}
\newcommand{\PiQPL}[1]{\Pi_{#1}\text{-}\logicClFont{SO}_2}
\newcommand{\SOQPLuniq}{\logicClFont{SO}_2^\textrm{u}}
\newcommand{\SigmaQPLuniq}[1]{\Sigma_{#1}\text{-}\logicClFont{SO}_2^\textrm{u}}
\newcommand{\PiQPLuniq}[1]{\Pi_{#1}\text{-}\logicClFont{SO}_2^\textrm{u}}

 \newcommand{\branch}[1]{\protect\ensuremath{\mathsf{branch}(#1)}}
 \newcommand{\store}[1]{\protect\ensuremath{\mathsf{store}(#1)}}
 \newcommand{\tree}[1]{\protect\ensuremath{\mathsf{tree}(#1)}}

\renewcommand{\vec}{\overline}
 
\newtheoremstyle{theorem}{\bigskipamount}{\medskipamount}{}{}{\bfseries}{.}{0.5em}{}
\newtheoremstyle{example}{\bigskipamount}{\medskipamount}{}{}{\bfseries}{:}{0.5em}{}

\theoremstyle{plain}

\newtheorem{theorem}{Theorem}[section]
\newtheorem*{theorem*}{Theorem}

\newaliascnt{lemma}{theorem}
\newaliascnt{corollary}{theorem}
\newaliascnt{definition}{theorem}
\newaliascnt{example}{theorem}

\newtheorem{lemma}[lemma]{Lemma}
\newtheorem{corollary}[corollary]{Corollary}
\newtheorem{proposition}[theorem]{Proposition}

\theoremstyle{definition}
\newtheorem{definition}[definition]{Definition}

\theoremstyle{example}
\newtheorem{example}[example]{Example}

\aliascntresetthe{lemma}
\aliascntresetthe{corollary}
\aliascntresetthe{definition}
\aliascntresetthe{example}

\crefname{theorem}{theorem}{theorems}
\crefname{lemma}{lemma}{lemmas}
\crefname{example}{example}{example}
\crefname{corollary}{corollary}{corollaries}
\crefname{definition}{definition}{definitions}
\crefname{example}{example}{examples}

\title{On Quantified Propositional Logics and the Exponential Time Hierarchy}
\author{Miika Hannula
\institute{Department of Computer Science \\ The University of Auckland}
\email{m.hannula@auckland.ac.nz}
\and
Juha Kontinen
\institute{Department of Mathematics and Statistics \\ University of Helsinki}
\email{juha.kontinen@helsinki.fi}
\and Martin Lück
\institute{Institut für Theoretische Informatik \\ Leibniz Universität Hannover}
\email{lueck@thi.uni-hannover.de}
\and
 Jonni Virtema
\institute{Department of Mathematics and Statistics \\ University of Helsinki}
\institute{Institut für Theoretische Informatik \\ Leibniz Universität Hannover}
\email{jonni.virtema@helsinki.fi}}
\def\titlerunning{On Quantified Propositional Logics and the Exponential Time Hierarchy}
\def\authorrunning{M. Hannula, J. Kontinen, M. Lück \& J. Virtema}

\def\copyrightholders{M. Hannula, J. Kontinen,\\\phantom{\copyright}~M. Lück \& J. Virtema}

\begin{document}
\maketitle

\begin{abstract}
We study quantified propositional logics from the complexity theoretic point of view. First we introduce
alternating dependency quantified boolean formulae ($\ADQBF$) which generalize both quantified and dependency quantified boolean formulae.
We show that the truth evaluation for
$\ADQBF$ is \AEXPPOLY-complete. We also identify fragments for which the problem is complete for the levels of the exponential hierarchy. Second we study propositional team-based logics. We show that $\DQBF$ formulae correspond naturally to quantified propositional dependence logic and present a general $\NEXPTIME$ upper bound for quantified propositional logic with a large class of generalized dependence atoms. Moreover we show \AEXPPOLY-completeness for extensions of propositional team logic with generalized dependence atoms.
\end{abstract}

\section{Introduction}

Deciding whether a given quantified propositional formula (qBf) is valid is a canonical $\PSPACE$-complete problem \cite{Stockmeyer:1973}.
\emph{Dependency quantified propositional formulae} (dqBf) introduced by Peterson et~al.~\cite{Peterson2001} are variants of qBfs for which the corresponding decision problem is $\NEXPTIME$-complete. Intuitively the rise of complexity stems from the fact that existential second-order quantification (existential quantification of Boolean functions) can be expressed in dqBf.

We present several logical formalisms, based on quantified propositional logic, that capture the concept of function quantification. We start from a variant of qBf where quantification happens on the level of Boolean functions in form of explicit syntactical objects. This \emph{second-order qBf}, introduced in \Cref{sec:soqbf}, captures the exponential hierarchy~\cite{lohrey_model-checking_2012,sqbf_report}.
In \Cref{sec:dqbf} we extend dqBf to incorporate universal quantification of Skolem functions, and show that second-order qBf can be translated to this novel formalism in logspace.
In Sections \ref{sec:team} and \ref{sec:atoms} we finally study \emph{dependence logic} and \emph{team logic} \cite{vaananen07} in the framework of qBf. We give efficient translations between these different formalisms and prove that they all capture the exponential hierarchy.

For a detailed exposition on dependence logics see the recent survey \cite{dukovo16}.
For the definition of the relevant complexity classes we follow the definition of \emph{alternating Turing machines} by Chandra, Kozen and Stockmeyer \cite{alternation}.
The class $\AEXPPOLY$ is the class of all problems decidable by alternating Turing machines in exponential time, \ie, $\bigO{2^{n^\bigO{1}}}$, and polynomially many alternations between existential and universal states. The classes $\SigmaE{k}$ and $\PiE{k}$ of the exponential hierarchy are similar but with at most $k$ alternations, where $k \in \mathbb{N}$, starting in an existential resp.\ universal state. The classes are closed both under $\leqpm$- and $\leqlogm$\emph{-reductions}. In this paper, if not specified otherwise, when we speak of \emph{reductions} we mean $\leqlogm$\emph{-reductions}.

\section{Second-order propositional logic}\label{sec:soqbf}

\emph{Second-order propositional logic} is obtained from usual qBf by shifting from quantification over proposition variables to quantification over Boolean functions. In this setting Boolean functions with arity $0$ correspond to propositional variables. Boolean functions with arity $\geq 1$ are called \emph{proper functions}.

For a formal definition let $\Phi$ be a set of function symbols. Every function symbol $f \in \Phi$ has its own well-defined arity $\arity{f}$. The syntax of $\SOQPL(\Phi)$ is given as follows:

\[
\phi \ddfn (\phi \wedge \phi) \mid \neg \phi \mid \exists f \phi \mid f(\underbrace{\phi, \ldots, \phi}_{n \text{ times}})\text{, where }f \in \Phi\text{ and }\arity{f} = n.
\]

The symbols $\forall$, $\lor$, $\imp$ and $\leftrightarrow$ are defined as the usual abbreviations.
We call this logic $\SOQPL$ as it essentially corresponds to second-order predicate logic $\SO$ restricted to the domain $\{0,1\}$. $\ESOQPL$ is the fragment of $\SOQPL$ where quantifiers $\exists$ for proper functions may occur only in the scope of even number of negations,  i.e., universal quantification of proper functions is disallowed.

\begin{definition}[$\SOQPL$ semantics]
The semantics of $\SOQPL(\Phi)$ is defined with assignments that map variables to Boolean functions: A $\Phi$-\emph{interpretation} $S$ is a map from $\Phi$ to Boolean functions, \ie, for any $f \in \Phi$ with arity $\arity{f} = n$, $S(f)\, \colon \{0,1\}^n \to \{0,1\}$ is an $n$-ary Boolean function.
Given an $\SOQPL(\Phi)$-formula $\phi$, write $\llbracket \phi \rrbracket_S$ for its valuation in $S$, which is defined as
\begin{alignat*}{2}
&\llbracket \phi \land \psi \rrbracket_S &&\dfn \llbracket \phi \rrbracket_S \cdot \llbracket \psi \rrbracket_S,\\
&\llbracket \neg \phi \rrbracket_S &&\dfn 1 - \llbracket \phi \rrbracket_S,\\
&\llbracket f(\phi_1,\ldots,\phi_n) \rrbracket_S &&\dfn S(f)(\llbracket \phi_1 \rrbracket_S, \ldots, \llbracket \phi_n \rrbracket_S)\text{, where }\arity{f} = n,\\
&\llbracket \exists f \phi\rrbracket_S &&\dfn \max\Set{\llbracket\phi\rrbracket_{S^f_F} | F \,\colon \{0,1\}^n \to \{0,1\}},
\end{alignat*}
where $S^f_F$ is the $\Phi$-interpretation \suchthat{} $S^f_F(f) = F$ and $S^f_F(g) = S(g)$ for $g \neq f$.
An $\SOQPL(\Phi)$-formula $\phi$ is \emph{valid} if $\llbracket \phi \rrbracket_S = 1$ for all $\Phi$-interpretations $S$. It is \emph{satisfiable} if there is at least one $S$ \suchthat $\llbracket \phi \rrbracket_S = 1.$ Finally, $\phi$ is \emph{true} if it contains no free variables and it is valid. If $\vec f = (f_1,\dots, f_n)$ is a tuple of function symbols, we sometimes write $\forall\vec f$ for $\forall f_1\dots \forall f_n$ and $\exists\vec f$ for $\exists f_1\dots \exists f_n$.

\end{definition}

In the following we drop $\Phi$ and just assume that it contains sufficiently many function symbols of any finite arity.

\begin{definition}
A second-order formula is \emph{simple} if functions have only propositions as arguments. It is in \emph{prenex form} if all quantifiers are at the beginning of the formula, and all proper functions are quantified before propositions, \ie, it is of the form
\[
\phi = \Game_1 f_1 \dots \Game_n f_n \Game'_1 x_1 \dots \Game'_m x_m \psi\text{,}
\]
where $n,m \geq 0$, $\{\Game_1, \ldots, \Game_n, \Game'_1, \ldots, \Game'_m\} \subseteq \{\exists, \forall\}$, the functions $f_1,\ldots,f_n$ have arity $>0$, the functions $x_1,\ldots,x_m$ have arity $0$, and $\psi$ is a quantifier-free propositional formula.

Write $\SigmaQPL{k}$ for the restriction of $\SOQPL$ to formulae of the form $\phi = \Game_1 \vec f_1 \; \ldots\; \Game_k \vec f_k \;\psi$, where $k \in \mathbb{N}$, $\Game_1 = \exists$, $\{\Game_2, \ldots, \Game_n\} \subseteq \{\exists, \forall\}$, each $f_i$ is a possibly empty tuple of proper function symbols, and $\psi$ is an $\SO_2$ formula in which all quantifications are of propositional variables. If $\Game_1 = \forall$, the corresponding fragment is called $\PiQPL{k}$.
The restriction to formulae where each function symbol $f$ occurs only with a fixed tuple $\vec{a}_f$ of arguments is denoted by the suffix $\cdot^\mathrm{u}$.
\end{definition}
Let $\LL$ be some logic. The problems $\TRUE(\LL)$, $\SAT(\LL)$, and  $\VAL(\LL)$ are defined as follows: Given a formula $\phi\in\LL$, decide whether the formula is true, satisfiable, or valid, respectively.

The restricted variant $\SOQPL$ of second-order logic $\SO$ is obviously decidable due to its finite domain. Moreover it captures exactly the levels of the exponential hierarchy.

\begin{proposition}[\cite{lohrey_model-checking_2012}]\label{thm:qbsf-completeness}
For any $k \geq 1$ the following problems restricted to prenex formulae are complete \wrt $\leqpm$:
$\TRUE(\SigmaQPL{k})$ is $\SigmaE{k}$-complete, $\TRUE(\PiQPL{k})$ is $\PiE{k}$-complete,
and $\TRUE(\SOQPL)$ is $\AEXPPOLY$-complete.
\end{proposition}

\section{Dependency quantified propositional formulae}\label{sec:dqbf}

In the previous section we considered second-order propositional logic. Now we turn to logics in which functions are quantified only implicitly in form of \emph{Skolem functions} of variables. Well-known such logics are \emph{dependency quantified propositional formulae} ($\DQBF$), but also \emph{independence-friendly logic} ($\IF$) by Hintikka and Sandu~\cite{hintikka89}. They have in common the syntactical property that Skolem functions are specified by denoting \emph{constraints} for quantified variables.
It is worth noting that we get the standard quantified propositional logic by restricting attention to formulae of $\SOQPL$ in which it is only allowed to quantify functions of arity $0$. Furthermore, $\DQBF$ correspond to the fragment  $\ESOQPLuniq$. In this section we introduce a generalization of $\DQBF$ that analogously corresponds to the full logic $\SOQPLuniq$.

We start by giving the definition of  $\DQBF$ and some required notation. For the definitions related to $\DQBF$, we follow Virtema~\cite{virtema14}.
 For a set $C$ of propositional variables, we denote by $\vec{c}$ the canonically ordered tuple of the variables in the set $C$. We refer to usual propositional assignments, in contrast to function assignments, by $s$ instead of $S$.

A formula that does not have any free variables is called \emph{closed} (or a  \emph{sentence}).  A \emph{simple qBf} is a closed qBf of the type
$
\phi \dfn \forall p_1 \cdots \forall p_{n}\exists q_1\cdots \exists q_m\theta,
$
where $\theta$ is a propositional formula and the propositional variables $p_i,q_j$ are all distinct. Any tuple $(C_1,\dots,C_m)$ such that $C_1,\dots,C_m\subseteq \{p_1,\dots,p_n\}$ is called a \emph{constraint} for $\phi$.
\begin{definition}\label{DQBF-semantics}
A simple qBf $\forall p_1 \cdots \forall p_{n}\exists q_1\cdots \exists q_m\theta$ is \emph{true under a constraint} $(C_1,\dots,C_m)$ if there exist functions $f_1,\dots,f_m$ with
$
f_i\colon \{0,1\}^{\lvert C_i\rvert}\to \{0,1\}
$
such that for each assignment $s\colon \{p_1,\dots,p_n\}\to\{0,1\}$,
$
s\big(q_1 \mapsto f_1(s(\vec{c}_1)), \dots, q_m \mapsto f_m(s(\vec{c}_m))\big)\models \theta.
$
\end{definition}

A \emph{dependency quantified propositional formula} is a pair $(\phi, \vec C)$, where $\phi$ is a simple quantified propositional formula and $\vec{C}$ is a constraint for $\phi$. We say that $(\phi, \vec C)$ is \emph{true} if $\phi$ is true under the constraint $\vec C$.
Let $\DQBF$ denote the set of all dependency quantified propositional formulae.
\begin{proposition}[{\cite[5.2.2]{Peterson2001}}]\label{TDQBF problem}
$\TRUE(\DQBF)$ is $\NEXPTIME$-complete problem w.r.t.\ $\leqlogm$.
\end{proposition}

We next introduce a novel variant of $\DQBF$ called \emph{alternating dependency quantified propositional formulae}  $\ADQBF$.
The syntax of  \emph{alternating quantified propositional formulae} extends the syntax of quantified propositional formulae with
  a new quantifier $\U$. The quantifier $\U$ is used to  express universal quantification of Skolem functions of  propositional symbols.

\begin{definition}
A \emph{simple $\Sigma_k$-alternating qBf} is a closed formula of the form 
\[
\phi \dfn \forall p_1 \cdots \forall p_{n} \, (\exists q^1_1\cdots \exists q^1_{j_1}) \, (\U q^2_1 \cdots \U q^2_{j_2}) \,  (\exists q^3_1 \cdots \exists q^3_{j_3}) \dots  (\, Q q^k_1 \cdots Q q^k_{k}) \,   \theta,
\]
where $Q\in\{\exists,\U\}$, $\theta$ is a propositional formula and the quantified propositional variables are all distinct.
Similarly, a \emph{simple $\Pi_k$-alternating qBf} is a closed formula of the form 
\[
\phi \dfn \forall p_1 \cdots \forall p_{n} \, (\U q^1_1\cdots \U q^1_{j_1}) \, (\exists q^2_1 \cdots \exists q^2_{j_2}) \,  (\U q^3_1 \cdots \U q^3_{j_3}) \dots  (\, Q q^k_1 \cdots Q q^k_{j_k}) \,   \theta.
\]
A \emph{simple alternating qBf} is a simple $\Sigma_i$-alternating or  $\Pi_i$-alternating qBf for some $i$.
Any tuple $(C^1_1,\dots,C^k_{j_k})$ such that $C^1_1,\dots,C^k_{j_k}\subseteq \{p_1,\dots,p_n\}$ is called a \emph{constraint} for $\phi$.
\end{definition}

\begin{definition}\label{def:semantics-adqbf}
Truth of a simple alternating qBf under a constraint $(C^1_1,\dots,C^k_{j_k})$  is defined by generalizing Definition \ref{DQBF-semantics} such that each  $\U q^l_i$ is interpreted  as universal quantification over (Skolem) functions  $f^l_i\colon \{0,1\}^{\lvert C^l _i\rvert}\to \{0,1\}$.

\end{definition}

\begin{example}
Let $\phi\dfn \forall p_1\forall p_2 \exists q_1 \U q_2\, \theta$ and $C\dfn ( \{p_2\}, \{p_1\} )$. Now $(\phi,C)$ is true iff there exists a function $f_1:\{0,1\} \to \{0,1\}$ \suchthat for all functions $f_2:\{0,1\} \to \{0,1\}$ it holds that for each assignment $s\colon \{p_1,p_2\}\to\{0,1\}$,
$s\big(q_1 \mapsto f_1(s(p_2)), q_2 \mapsto f_2(s(p_1))\big)\models \theta$.
\end{example}

\begin{example}The formula $\forall \tuple x (\U y \exists z) \neg y \leftrightarrow z$ under the constraint $(\{\tuple x\},\{\tuple x\})$ expresses that every $|\tuple x|$-ary Boolean function has a negation.
\end{example}

\begin{definition}
The set $\ADQBF$ is the set of all pairs $(\phi, \tuple C)$ where $\phi$ is a simple alternating qBf and $\tuple C$ is a constraint of $\phi$.
The set $\Sigma_k$-$\ADQBF$ ($\Pi_k$-$\ADQBF$) is then the subset of $\ADQBF$ where $\phi$ is $\Sigma_k$-alternating ($\Pi_k$-alternating).
\end{definition}

\begin{lemma}\label{thm:dqbf-membership}
For all $k \geq 1$ it holds that $\TRUE(\Sigma_k\text{-}\ADQBF) \in \SigmaE{k}$, $\TRUE(\Pi_k\text{-}\ADQBF) \in \PiE{k}$, and $\TRUE(\ADQBF) \in \AEXPPOLY$.
\end{lemma}
\begin{proof}
	We give a brute-force algorithm. Let the universal quantified prefix of the given qBf be $\forall p_1 \ldots \forall p_n$. For every $\exists$-quantified block $\exists q_1 \ldots \exists q_{j}$ with constraints $C_1,\ldots,C_{j}$, existentially guess and write down a Boolean function from the variables $C_i \subseteq \{p_1,\ldots,p_n\}$ for every $q_i$. For every $\U$-quantified block, switch to universal branching and do the same. The quantifier-free part can then be evaluated in deterministic exponential time for every possible assignment to $p_1,\ldots,p_n$. The algorithm runs in exponential time and its alternations are bounded by the alternations of $\exists$ and $\U$ quantifiers in the formula.
\end{proof}

For the hardness direction we first show how the uniqueness property can be obtained for arbitrary $\SOQPL$-formulae by introducing additional function symbols.
The following lemma will be needed in the sequel (see, \eg, Väänänen \cite{vaananen07}).
\begin{lemma}\label{thm:quantifier-swap}
Let $S$ be an $\SOQPL$ interpretation, $x$ a propositional variable, $f$ and $f'$ function variables with $\arity{f'} = \arity{f} + 1$, and $\phi(x, f)$ an $\SOQPL$-formula in which $x$ and $f$ occur only as free variables and in which $f'$ does not occur. Then it holds that $S \models \forall x\exists f \phi \Leftrightarrow S \models \exists f'\, \forall x \phi'$ and $S \models \exists x\forall f \phi \Leftrightarrow S \models \forall f' \exists x \phi'$,
where $\phi'$ is obtained from $\phi$ by replacing $f(\vec{y})$ with $f'(x,\vec{y})$.
\end{lemma}

\begin{lemma}\label{thm:simple-prenex-so}
Every $\SigmaQPL{k}$-sentence \emph{(}$\PiQPL{k}$-sentence, $\SOQPL$-sentence\emph{)} $\phi$ can be translated to an equivalent simple prenex $\SigmaQPLuniq{k}$-sentence \emph{(}$\PiQPLuniq{k}$-sentence, $\SOQPLuniq$-sentence\emph{)} $\psi$ in polynomial time.
\end{lemma}
\begin{proof}
First we prove that $\phi$ can be transformed into a simple formula in polynomial time. Whenever a subformula $\xi = f(\psi_1,\ldots,\psi_i,\ldots,\psi_n)$ occurs and $\psi_i$ is not a proposition, then replace $\xi$ by $\forall b \big( (b\leftrightarrow \psi_i) \rightarrow f(\psi_1,\ldots,b,\ldots,\psi_n) \big)$, where $b$ is a new proposition symbol.

Then by the usual translation move all quantifiers to the beginning of the formula. Swap the order of the quantifiers according to \Cref{thm:quantifier-swap} until all quantified proper function symbols precede the quantified propositions. Such obtained $\phi$ is simple and of the form $\Game_1 f_1  \dots \Game_n f_n \forall p_1 \exists q_1 \dots \forall p_m \exists q_m \theta$ where $\theta$ is quantifier-free, $\{\Game_1,\ldots,\Game_n\} \subseteq \{\exists, \forall\}$, and $f_1,\ldots,f_n$ are the only proper functions that occur in $\phi$.

Finally we \enquote{split} the quantified function symbols in $\phi$ \suchthat every proper function symbol occurs with exactly one fixed argument tuple. Let
$\chi=\exists f \Game_1 g_1 \dots \Game_k g_k \Game'_1 p_1 \dots \Game'_m p_m \,\theta(f(\tuple x_1), \ldots, f(\tuple x_n))$
where $\theta$ is quantifier-free be a subformula of $\phi$, meaning that $f$ occurs in $\theta$ at $n$ different positions with $n$ (possibly different) argument tuples $\tuple x_1, \ldots, \tuple x_n$. In what follows, $f_1 \ldots f_n$ and $\tuple y_1, \ldots, \tuple y_n$ are assumed to be distinct and fresh. Then $\phi$ is equivalent to the formula obtained from $\phi$ by substituting $\chi$ by the formula
$\exists f_1 \ldots \exists f_n \Game_1 g_1 \dots \Game_k g_k \Game'_1 p_1 \dots \Game'_m p_m\, \forall \,\tuple y_1 \ldots \,\forall\, \tuple y_n\,(\psi_1 \land \psi_2)$,
where $\psi_1 \dfn \bigwedge_{i = 1}^{n - 1} \big((\tuple y_i \leftrightarrow \tuple y_{i+1}) \imp (f_i(\tuple y_i) \leftrightarrow f_{i+1}(\tuple y_{i+1}))\big)$ ensures that the functions $f_1, \ldots, f_n$ are all the same, and $\psi_2 \dfn \left(\bigwedge_{i = 1}^n \tuple x_i \leftrightarrow \tuple y_i\right) \imp \Big(\theta(f_1(\tuple y_1), \ldots, f_n(\tuple y_n))\Big)$ simulates $\theta(f_1(\tuple x_1), \ldots, f_n(\tuple x_n))$.
The \enquote{split} of universal quantifiers is done analogously. 
Clearly $\phi$ remains simple and in prenex form.

The steps introduced  above do not add additional alternations of function quantifiers, hence the resulting formula is now an $\SigmaQPLuniq{k}$ resp.\ $\PiQPLuniq{k}$ resp. $\SOQPLuniq$ sentence.
\end{proof}

\begin{theorem}\label{thm:odd-k-dqbf-hardness}
Let $k \geq 1$. For odd $k$ the problem $\TRUE(\Sigma_k\text{-}\ADQBF)$ is $\SigmaE{k}$-complete. For even $k$ the problem $\TRUE(\Pi_k\text{-}\ADQBF)$ is $\PiE{k}$-complete.
The problem $\TRUE(\ADQBF)$ is $\AEXPPOLY$-complete.
\end{theorem}
\begin{proof}The membership was shown in \Cref{thm:dqbf-membership}.
For the hardness we start with the problem $\TRUE(\Sigma_k\text{-}\ADQBF)$. We give a reduction from $\TRUE(\SigmaQPLuniq{k})$ which is by Proposition \ref{thm:qbsf-completeness} and Lemma \ref{thm:simple-prenex-so} $\leqpm$-complete for $\SigmaE{k}$. Let $\phi \dfn Q_1 \vec{f_1} Q_2 \vec{f_2} \ldots \exists \vec{f_k} \forall p_1 \exists q_1 \ldots \forall p_n \exists q_n \psi$, where $\psi$ is quantifier-free, be a simple prenex $\SigmaQPLuniq{k}$-sentence. Note that $Q_k = \exists$ since $k$ is odd.
For each function symbol $f_i$ that occurs in $\psi$, let $(a^i_1, \ldots, a^i_{m_i})$ denote the unique tuple that occurs as an argument of $f_i$.
Each of these functions with arguments can be simulated by a single constrained propositional variable; a problem in this translation is however that some $a^i_j$ may be existentially quantified and thus not part of the $p_1, \ldots, p_n$.
However, this problem can be easily solved by introducing fresh universally quantified propositional variables:

Assume that $\psi$ is in negation normal form. Any subformula $f_i(\tuple a)$ is replaced by $\forall \tuple r \big( (\tuple a \leftrightarrow \tuple r) \imp f_i(\tuple r)\big)$, where $\tuple r = r_1, \ldots, r_{\arity{f_i}}$ are fresh distinct variables. Analogously, $\neg f_i(\tuple a)$ is replaced by $\forall \tuple r \big( (\tuple a \leftrightarrow \tuple r) \imp \neg f_i(\tuple r)\big)$.  Clearly the such obtained sentence can be transformed to prenex form by just moving all the freshly introduced quantifiers to the right end of the quantifier prefix. The such obtained sentence is equivalent to $\phi$.
Thus we may assume \wloss that if $f_i(a^i_1, \ldots, a^i_{m_i})$ occurs in $\phi$, then $\{a^i_1, \ldots, a^i_{m_i}\} \subseteq \{ p_1, \ldots, p_n \}$.

For the reduction to $\ADQBF$ now just consider $\forall p_1 \ldots \forall p_n$ as universal quantified variables. For every $f_i(a^i_1, \ldots, a^i_{m_i})$ introduce a quantified variable $\exists f_i$ with constraint $\{ a^i_1, \ldots, a^i_{m_i} \}$ if $f_i$ is existentially quantified, and introduce $\U f_i$ with the same constraint otherwise. For every $q_i$ introduce a quantified variable $\exists q_i$ with constraint $\{p_1, \ldots, p_{i}\}$. Let $\psi'$ denote the formula obtained from $\psi$ by substituting, for each $i$, $f_i(a^i_1, \ldots, a^i_{m_i})$ by $f_i$. Thus the resulting $\ADQBF$ has the form
\(
\phi' \dfn \forall p_1 \ldots \forall p_n \; (\exists \vec{f_1}) \; (\U \vec{f_2}) \; \ldots \; (\exists \vec{f_k} \, \exists q_1 \,  \ldots \, \exists q_n) \; \psi'.
\)
Since the final function quantifier $Q_k$ was existential, we can merge the functions $\vec{f_k}$ and the quantified propositions $q_1\ldots q_n$ to a single existentially quantified block in $\ADQBF$.
So $\phi'$ is in $\Sigma_k$-$\ADQBF$, and by definition of $\ADQBF$, $\phi'$ is true under the constraint $C$ (where $C$ is constructed as above) if and only if $\phi \in \TRUE(\SigmaQPLuniq{k})$.

\smallskip

For general $\SOQPLuniq$ formulae we can again assume that the last function quantifier is existential. The same holds for $\PiQPLuniq{k}$ formulae if $k$ is even.
In these cases a similar proof yields a reduction to $\ADQBF$ resp.\ $\Pi_k$-$\ADQBF$.
\end{proof}

\begin{theorem}\label{thm:sigma-k-collapse}
Let $k \geq 2$. For even $k$ the problem $\TRUE(\Sigma_k\text{-}\ADQBF)$ is $\SigmaE{k-1}$-complete.
For odd $k$ the problem $\TRUE(\Pi_k\text{-}\ADQBF)$ is $\PiE{k-1}$-complete.
\end{theorem}
\begin{proof}
The hardness results follow from \Cref{thm:odd-k-dqbf-hardness}. For inclusion, we prove the case for $\Sigma_k\text{-}\ADQBF$. We give a $\leqlogm$-reduction from $\TRUE(\Sigma_k\text{-}\ADQBF)$ to $\TRUE(\Sigma_{k-1}\text{-}\ADQBF)$. The result then follows from  \Cref{thm:odd-k-dqbf-hardness}. The case for $\Pi_k\text{-}\ADQBF$ is analogous. Consider a formula
\(
\phi \dfn \forall p_1 \cdots \forall p_{n} \, (\exists q^1_1\cdots \exists q^1_{j_1}) \, (\U q^2_1 \cdots \U q^2_{j_2}) \, \dots  (\, \U q^k_1 \cdots \U q^k_{j_k}) \,   \theta
\)
and a constraint $C=(C^1_1, \ldots, C^k_{j_k})$.
We claim that $(\phi,C)$ is equivalent to
\[
\phi' \dfn \forall p_1 \cdots \forall p_{n} \forall q^k_1 \cdots \forall q^k_{j_k} \, (\exists q^1_1\cdots \exists q^1_{j_1}) \, (\U q^2_1 \cdots \U q^2_{j_2}) \, \dots  (\exists q^{k-1}_1 \cdots \exists q^{k-1}_{j_{k-1}}) \, \theta
\]
under the constraint $C'=(C^1_1, \ldots, C^{k-1}_{j_{k-1}})$.

By definition $(\phi,C)$ is true if and only if for all extensions of the tuple of quantified Skolem functions $f^1_1, \ldots, f^{k-1}_{j_{k-1}}$ (some of which are existentially/universally quantified) and for all extensions of the Skolem functions $f^k_i$ it holds that:
\begin{equation}\label{eq:constraint1}
\forall t \in T: t\models \theta, \text{ where } T \dfn \{ s\big(q^1_1 \mapsto f^1_1(\overline{c}^1_1), \ldots, q^k_{j_k} \mapsto f^k_{j_k}(\overline{c}^k_{j_k})\big) \mid s:\{p_1, \ldots, p_n\} \to \{0,1\} \}
\end{equation}
Note that, in fact,  $T$ is the set of all expansions of assignments $s(q^1_1 \mapsto f^1_1(\overline{c}^1_1), \ldots, q^{k-1}_{j_{k-1}} \mapsto f^{k-1}_{j_{k-1}}(\overline{c}^{k-1}_{j_{k-1}}))$,  $s:\{p_1, \ldots, p_n\} \to \{0,1\}$, into domain $\{p_1, \ldots, p_n, q^1_1, \ldots, q^{k}_{j_{k}}\}$.
Thus \eqref{eq:constraint1} can be equivalently written as
\begin{align}\label{eq:constraint2}
\forall t \in T: t\models \theta, \text{ where } T \dfn \{ s\big(q^1_1 \mapsto f^1_1(\overline{c}^1_1)&, \ldots, q^{k-1}_{j_{k-1}} \mapsto f^{k-1}_{j_{k-1}}(\overline{c}^{k-1}_{j_{k-1}})\big) \mid \\
&s:\{p_1, \ldots, p_n, q^k_1, \ldots, q^{k}_{j_{k}}\} \to \{0,1\} \} \notag
\end{align}
Now note that as the constraints and quantifiers for $q^1_1, \ldots, q^{k-1}_{j_{k-1}}$ are exactly the same in $C$ and in $C'$ each extension of the tuple of quantified Skolem functions $f^1_1, \ldots, f^{k-1}_{j_{k-1}}$ in the evaluation of $(\phi,C)$ can be directly interpreted in  $(\phi,C')$, and vice versa. From this together with the equivalence of \eqref{eq:constraint1} and \eqref{eq:constraint2}, we conclude that  $(\phi,C)$ and  $(\phi,C')$ are equivalent.
\end{proof}

Using the translation from $\SOQPLuniq$ to $\ADQBF$ introduced in the proof of \Cref{thm:odd-k-dqbf-hardness}, we obtain the following corollary.

\begin{corollary}\label{thm:so-sentence-to-adqbf-ptime}
For every $\SOQPL$-sentence $\phi$ there is a polynomial time computable $\ADQBF$-instance $(\psi, \tuple C)$ which is true iff $\phi$ is true.
For every $\ESOQPL$-sentence $\phi$ there is a polynomial time computable $\DQBF$-instance $(\psi, \tuple C)$ which is true iff $\phi$ is true.
\end{corollary}

\section{Quantified propositional logics with team semantics}\label{sec:team}
The study of propositional logics with team semantics  has so far concentrated on extensions of propositional logics with different dependency notions such as functional dependence, independence and inclusion.
Here we extend the perspective to quantified propositional logics.

\subsection{Basic notions and results}
In the team semantics context it is usual to consider assignments over  finite sets of proposition symbols. We begin by fixing some notation.
Let $D$ be a finite, possibly empty set of proposition symbols. 
A set $X$ of assignments $s\colon D\to \{0,1\}$ is called a \emph{team}. The set $D$ is the \emph{domain} $\Dom{X}$ of $X$.
We denote by $2^D$ the set of \emph{all assignments} $s\colon D\to \{0,1\}$. If $\vec p =(p_1, \ldots ,p_n)$ is a tuple of propositions and $s$ is an assignment, we write $s(\vec p)$ for $\left(s(p_1),\dots,s(p_n)\right)$. If $b\in \{0,1\}$ and $s$ is an assignment with domain $D$, we let $s(q\mapsto b)$ denote the assignment with domain $D\cup\{q\}$ defined as follows: $s(q\mapsto b)(p)=b$ if $p=q$ and  $s(q\mapsto b)(p)=s(p)$ if $p\not=q$.

\medskip

Let $X$ be a team. A function $F:X\to \{\{0\}, \{1\}, \{0,1\}\}$ is called a \emph{supplementing function} of $X$. Supplementing functions are used for giving semantics for existential quantifiers. For a proposition symbol $p$, we define $X[F/p] \dfn \{s(p\mapsto b) \mid s\in X, b\in F(s) \}$. We say that $X[F/p]$ is a \emph{supplemented team}  of $X$ in $p$.

For $A\subseteq \{0,1\}$ we define $X[A/p] \dfn \{s(p\mapsto b) \mid s\in X, b\in A \}$. The team $X[\{0,1\}/p]$ is the \emph{duplicating team} of $X$ in $p$. Duplicating teams are used to give semantics for universal quantifiers.

Let $\Phi$ be a set of proposition symbols. The syntax of \emph{quantified propositional team logic} $\QPTL(\Phi)$ is given by the following grammar:
\[
\phi \ddfn p\mid \neg p \mid (\phi \wedge \phi) \mid (\phi \vee \phi) \mid \negg \phi \mid \forall p \,\phi \mid \exists p\, \phi , \text{ where $p\in\Phi$,}
\]

Its quantifier-free fragment is called \emph{propositional team logic} $\PTL(\Phi)$, similar to the first-order team logic $\TL$ by Väänänen \cite{vaananen07}.
Likewise its $\negg$-free fragment is called \emph{quantified propositional logic} $\QPL(\Phi)$.
The usual \emph{propositional logic} $\PL(\Phi)$ is then just the quantifier-free fragment of $\QPL(\Phi)$.

Let us denote by $\Prop{\phi}$ the set of proposition symbols that occur in $\phi$, and by $\Fr{\phi}$ the set of proposition symbols that occur free in $\phi$. We sometimes write $\phi(p_1, \ldots ,p_n)$ to denote that $\phi$ is a formula whose free proposition symbols are in $\{p_1, \ldots ,p_n\}$. A formula in which no proposition symbol occurs free is called a \emph{sentence}.
We denote by $\modelsPL$ the ordinary satisfaction relation of quantified propositional logic defined via assignments in the  standard way. Next we give team semantics for quantified propositional logic. The semantics for the quantifiers follow the corresponding definitions of first-order team semantics (as quantified propositional logic can be seen as  first-order logic over domain $\{0,1\}$).
\begin{definition}[Lax team semantics]
Let $\Phi$ be a set of atomic propositions and let $X$ be a team. The satisfaction relation $X\models \phi$ for $\phi \in \QPTL(\Phi)$ is defined as follows.
\begin{align*}
X\models p  \quad\Leftrightarrow\quad& \forall s\in X: s(p)=1 \;\text{ if } p \in \Phi\text{.} \\
X\models \neg p \quad\Leftrightarrow\quad& \forall s\in X: s(p)=0 \;\text{ if } p \in \Phi\text{.}\\
X\models (\phi\land\psi) \quad\Leftrightarrow\quad& X\models\phi \text{ and } X\models\psi.\\
X\models (\phi\lor\psi) \quad\Leftrightarrow\quad& Y\models\phi \text{ and }
Z\models\psi,
\text{ for some $Y,Z$ such that $Y\cup Z= X$}.\\
X\models \negg \phi \quad\Leftrightarrow\quad& X \not\models\phi.\\
X\models \exists p\, \phi  \quad\Leftrightarrow\quad& X[F/p]\models \phi \,\text{ for some function $F:X\to \{\{0\}, \{1\},\{0,1\}\}$}. \\
X\models \forall p\, \phi  \quad\Leftrightarrow\quad& X[\{0,1\}/p]\models \phi.
\end{align*}
We say that a sentence $\phi$ is \emph{true} if $\{\emptyset\}\models \phi$, i.e., if the team with just the empty assignment satisfies $\phi$.
\end{definition}

The next proposition shows that the team semantics and the ordinary semantics for $\QPL$-formulae coincide.
\begin{proposition}[Flatness property \cite{vaananen07}]\label{PLflat}
Let $\phi$ be a formula of quantified propositional logic and let $X$ be a propositional team. Then
\(
X\models \phi \;\text{ iff }\; \forall s\in X: s\modelsPL \phi.
\)
\end{proposition}

The syntax of \emph{quantified propositional dependence logic} $\QPD(\Phi)$ is obtained by extending the syntax of $\QPL(\Phi)$ by the following grammar rule for each $n\in\mathbb{N}$:
\[
\phi \ddfn \dep{p_1,\dots,p_n,q},\text{ where }p_1,\dots,p_n,q\in\Phi.
\]
The  meaning of the \emph{propositional dependence atom} $\dep{p_1,\dots,p_n,q}$ is that the truth value of the proposition symbol $q$ is functionally determined by the truth values of the proposition symbols $p_1,\dots,p_n$.
The semantics for the atoms is defined as follows:
$X\models \dep{p_1,\dots,p_n,q}$ iff for all $s,t\in X:s(p_1)=t(p_1), \dots, s(p_n)=t(p_n)$ implies $s(q)=t(q)$.

The next well-known result is proved in the same way as the analogous result for first-order dependence logic \cite{vaananen07}.

\begin{proposition}[Downwards closure]\label{dcprop}
Let $\phi$ be a $\QPD$-formula and let $Y\subseteq X$ be propositional teams. Then $X\models \phi$ implies $Y\models \phi$.
\end{proposition}

In this article we study also a variant of $\QPD$ obtained by replacing dependence atoms by the  so-called \emph{inclusion atoms}.  The  syntax of \emph{quantified propositional inclusion logic} $\QPLInc(\Phi)$ is obtained by extending the syntax of $\QPL(\Phi)$ by the grammar rule
\(
\phi \ddfn (p_1,\ldots,p_n) \subseteq (q_1,\ldots,q_n)
\)
for every $n \geq 0$, where $p_1,\ldots,p_n,q_1,\ldots,q_n \in \Phi$.
The semantics for propositional inclusion atoms is defined as follows:
\(
X\models {\vec p}\subseteq {\vec q}\text{ iff }\forall s\in X \,\exists\, t\in X :s(\vec p)=t(\vec q).
\)

It is easy to check that $\QPLInc$  is not downward closed (cf. Proposition \ref{dcprop}). However, analogously to FO-inclusion-logic \cite{Galliani12}, $\QPLInc$ is closed \wrt unions:

\begin{proposition}[Closure under unions]\label{closureunions}
Let $\phi \in \QPLInc$ and let $ X_i$, for $i\in I$, be teams. Suppose that $X_i\models \phi$ for each $i\in I$. Then $\bigcup_{i\in I} X_i \models \phi$.
\end{proposition}

\begin{definition}
Let $\LL$ be a propositional logic with team semantics. Recall that a sentence $\phi\in\LL$ is \emph{true} if $\{\emptyset\}\models \phi$. A formula $\phi\in\LL$ is \emph{satisfiable} if there exists a non-empty team $X$ such that $X\models \phi$. A formula $\phi\in\LL$ is \emph{valid} if $X\models \phi$ holds for all teams $X$ such that the proposition symbols in $\fr(\phi)$  are in the domain of $X$. The problems $\TRUE(\LL)$, $\SAT(\LL)$, and  $\VAL(\LL)$ are defined in the obvious way: Given a formula $\phi\in\LL$, decide whether the formula is true, satisfiable or valid, respectively.
\end{definition}

The following results for $\PLInc$ and $\MInc$ are implicitly shown by Hella et~al.~\cite{hkmv15}. They state the results using $\PSPACE$-reductions, but in fact their reductions run in polynomial time.

\begin{proposition}[\cite{hkmv15,virtema14}]\label{SatPLInc}
$\SAT(\PLInc)$ and  $\SAT(\MInc)$ are $\EXPTIME$-complete \wrt $\leqpm$-reductions.
$\VAL(\PD)$ is $\NEXPTIME$-complete \wrt $\leqlogm$-reductions.
\end{proposition}

The following lemma is a direct consequence of a result of Galliani~et~al.~\cite[Lemma 14]{GallianiHK13}, where an analogous claim  is proven in the first-order setting  over structures with universe size at least $2$. The result follows by the obvious back-and-forth translations between propositional logic and first-order logic where truth of a propositional formula is replaced with satisfaction by the  first-order structure that has universe $\{0,1\}$ and two constants interpreted as $0$ and $1$.

\begin{lemma}[\cite{GallianiHK13}]\label{normal form}
Any formula $\phi$ in $\LL$, where $\LL\in \{\QPD,\QPLInc\}$, is logically equivalent to a polynomial size formula $\Game_1 p_1 \ldots \Game_k p_k\psi$ in $\LL$ where $\psi$ is quantifier-free and $\{\Game_1, \ldots, \Game_k\}\subseteq\{\exists, \forall \}$ for $i=1, \ldots ,n$.
\end{lemma}

\subsection{Complexity of quantified propositional logics}

In this section we consider the complexity of quantified propositional dependence and inclusion logic. In the latter case, we reduce the problem to the satisfiability problem of \emph{modal inclusion logic}, $\MInc$, as defined by Hella et~al.~\cite{hkmv15}.

\begin{proposition}\label{TrueQPD}
$\TRUE(\QPD)$ is $\NEXPTIME$-complete \wrt $\leqlogm$-reductions.
\end{proposition}
\begin{proof}
We show a reduction from $\VAL(\PD)$ to $\TRUE(\QPD)$. By Proposition \ref{SatPLInc}, the former is $\NEXPTIME$-hard and thus the latter is as well.
Let $\phi$ be a $\PD$-formula and let $\vec{p}$ be the tuple of proposition symbols that occur in $\phi$. Note first that, since $\PD$ is downward closed, it follows that $\phi$ is valid if and only if $2^{\vec{p}}\models \phi$, where $2^{\vec{p}}$ is the team that contains exactly all propositional assignments with domain $\vec{p}$. Thus it follows that the $\PD$-formula $\phi$ is valid if and only if the $\QPD$-formula $\forall\, \vec{p}\, \phi$ is true.

The fact that $\TRUE(\QPD)$ is in $\NEXPTIME$ follows from the obvious brute force algorithm that uses non-determinism to guess the witnessing teams for existential quantifiers and disjunctions.
\end{proof}

\begin{theorem}\label{thm:qplinc}
$\TRUE(\QPLInc)$ is $\EXPTIME$-complete \wrt $\leqpm$-reductions.
\end{theorem}
\begin{proof}
We give a $\leqpm$-reduction from $\SAT(\PLInc)$ to $\TRUE(\QPLInc)$. Since, by Proposition \ref{SatPLInc}, $\SAT(\PLInc)$ is $\EXPTIME$-hard under $\leqpm$-reductions, it follows that  $\TRUE(\QPLInc)$ is as well.
Let $\phi$ be a formula of $\PLInc$ and let $\vec{p}$ be the tuple of proposition symbols that occur in $\phi$. Clearly there exists a nonempty propositional team $X$ such that $X\models \phi$ if and only if $\{\emptyset\}\models \exists \vec{p}\, \phi$.

We will next show that $\TRUE(\QPLInc)$ is in $\EXPTIME$. We do this via a polynomial time translation $\phi\mapsto \phi^*$ from $\QPLInc$ to $\MInc$.
The translation is designed such that $\phi$ is true if and only if $\phi^*$ is satisfied by a non-empty team in a Kripke structure. Since, by Proposition \ref{SatPLInc}, $\SAT(\MInc)$ is in $\EXPTIME$, it follows that $\TRUE(\QPLInc)$ is as well. In our construction, the idea is that points in a $Kripke$ model will correspond to propositional assignments, and existential and universal quantifiers are simulated by diamonds and boxes, respectively.

First we will enforce a binary (assignment) tree in our structure. Branching in the tree will correspond to quantification of proposition variables.
The binary tree is forced in the standard way by modal formulae:
The formula $\branch{p_i}\dfn \Diamond p_i\land\Diamond\lnot p_i$ forces that there are $\ge 2$ successor states which disagree on a proposition $p_i$.
The formula $\store{p_i}\dfn (p_i\land\Box p_i)\lor(\lnot p_i\land\Box\lnot p_i)$ is used to propagate chosen values for $p_i$ to successors in the tree.
Now define
\begin{align*}
\tree{p,n}\dfn \branch{p_1}\land\bigwedge_{i=1}^{n-1}\Box^i\Bigl(\branch{p_{i+1}}\land\bigwedge_{j=1}^i\store{p_j}\Bigr),
\end{align*}
where $\Box^i\phi\dfn\overbrace{\Box\cdots\Box}^{i\text{ many}}\phi$. The formula  $\tree{p,n}$ forces a complete binary assignment tree of depth $n$ for proposition symbols $p_1,\dots,p_n$. Notice that $\tree{p,n}$ is an $\ML$-formula and hence has the flatness property, analogously to Proposition \ref{PLflat} \cite{vaananen_modal_2008}. When $\phi$ is a $\QPLInc$-formula, we denote by $\phi'$ the $\MInc$ formula that is obtained from  $\phi$ by substituting each existential quantifier $\exists p$ by $\Diamond$ and each universal quantifier $\forall p$ by $\Box$.

We are now ready to state our reduction. Let $\phi$ be an arbitrary $\QPLInc$-formula in the normal form of Lemma \ref{normal form}. \Wloss we may assume that $\phi = \Game_1 p_1\dots \Game_n p_n \psi$, where $\{\Game_1,\ldots,\Game_n\}\subseteq\{ \exists, \forall\}$ and $\psi$ is quantifier-free. Define
\(
\phi^* \dfn  \tree{p,n} \wedge \phi'.
\)
It is straightforward to check that, indeed, $\phi$ is true if and only if $\phi^*$ is satisfiable.
\end{proof}

\subsection{Propositional team logic and \texorpdfstring{$\ADQBF$}{ADQBF}}
In \cite{piil_complexity_2015} it was established that the validity and satisfiability problem of  $\PTL$ extended with either inclusion  or independence atom is complete for $\AEXPPOLY$. In the extended version of the paper it is shown that, in fact, this holds this holds already for $\PTL$. Here we generalize this result by establishing connections between fragments of team-based logics and $\ADQBF$.

First observe that sentences of $\ADQBF$ can be equivalently interpreted as sentences of $\QPTL$ extended with dependence atoms, denoted by  $\QPTL(\depop{})$. This translation is analogous to the translation from $\SO$ to first-order team logic  (see \cite{kontinen2011team,nurmi09}).
Let
\[
\phi \dfn \forall p_1 \cdots \forall p_{n} \, (\exists q^1_1\cdots \exists q^1_{j_1}) \, (\U q^2_1 \cdots \U q^2_{j_2}) \,  (\exists q^3_1 \cdots \exists q^3_{j_3}) \dots  (\, Q q^k_1 \cdots Q q^k_{j_k}) \,   \theta,
\]
be a simple alternating qBf with constraints $(C^1_1,\dots,C^k_{j_k})$. Recall that for a set of variables $C$, we denote by $\vec{c}$ the canonically ordered tuple consisting of the variables in $C$. Let $\phi^*$ denote the following $\QPTL(\depop{})$-sentence:
\begin{align}
\forall p_1 \cdots \forall p_{n}& \, (\exists q^1_1\cdots \exists q^1_{j_1}) \, (\U q^2_1 \cdots \U q^2_{j_2}) \,  (\exists q^3_1 \cdots \exists q^3_{j_3}) \dots  (\, Q q^k_1 \cdots Q q^k_{j_k})  \label{eq:trans1}\\
&\negg\Bigg[  \negg (p\land \neg p) \land \bigwedge_{\overset{1\leq i \leq k}{\overset{i \text{ is even}}{1 \leq l \leq j_i}}} \dep{\overline{c}^i_l, y^i_l} \Bigg] \lor \Bigg[ \Big( \bigwedge_{\overset{1\leq i \leq k}{\overset{i \text{ is odd}}{1 \leq l \leq j_i}}} \dep{\overline{c}^i_l, y^i_l}\Big) \land \theta \Bigg] \notag
\end{align}
Above  the quantifier $\U q$  is treated  as a shorthand for the expression $\negg \exists q \negg$.\footnote{The syntax is the same as in \Cref{def:semantics-adqbf}. However, for $\ADQBF$, $\U$ refers to the universal quantification of Skolem functions, while in team semantics, it refers to the universal quantification of supplementing functions. These notions are not the same, but easily translatable into each other, as we show.} It is straightforward to check that $\phi$ is true under the constraint $(C^1_1,\dots,C^k_{j_k})$ if and only if $\phi^*$ is true.
Thus we obtain fragments of  $\QPTL(\depop{})$ that express complete problems for levels of the exponential hierarchy, see Theorem \ref{thm:odd-k-dqbf-hardness}.
For $k=1$, we obtain a translation from $\DQBF$ to $\QPD$.
It is noteworthy that, in fact, by the above translation we obtain a close connection between the classes $\SigmaE{k}$ and $\PiE{k}$, and the fragment of $\QPTL(\depop{})$ of sentences with $\leq k$ nested $\negg$s  ($\degg{\phi}$); formally defined as follows:
\begin{align*}
& \degg{\forall p\phi} \dfn  \degg{\exists p \phi} \dfn  \degg{\phi}, \,  \degg{\phi\lor \psi}\dfn \degg{\phi\land \psi}\dfn \max\{\degg{\phi},  \degg{\psi}\},\\
&\degg{\neg \phi} \dfn \degg{\phi}, \, \degg{\negg \phi} \dfn \degg{\phi}+1, \, \degg{\dep{\vec{p},q}}\dfn \degg{p} \dfn 0.
\end{align*}
Note that the relationship given by this translation is not strict. It is easy to show, by a brute-force algorithm, that $\TRUE(\LL)$ is in $\SigmaE{k+1}$, where $\LL$ is the fragment of $\QPTL(\depop{})$ with formulae with $\degg{\phi}\leq k$. Moreover, from the above translation together with Theorem \ref{thm:odd-k-dqbf-hardness} we obtain hardness for $\SigmaE{k-2}$.

\begin{proposition}\label{thm:adqbf-to-ptl}
Every $\ADQBF$-instance \emph{(}$\DQBF$-instance\emph{)} $(\psi, \tuple C)$ can be translated in polynomial time to a $\QPTL(\depop{})$-sentence \emph{(}$\QPD$-sentence\emph{)} $\phi$ \suchthat $(\psi,\tuple C)$ is true iff $\phi$ is true.
\end{proposition}

Using the ideas of \cite{piil_complexity_2015}, we may eliminate the quantifiers in \eqref{eq:trans1} and relate the truth of $\phi$ and $\phi^*$ with the satisfiability of a certain formula of $\PTL$ extended with dependence atoms, denoted by $\PTL(\depop{})$. Define a shorthand $\max(p_1,\dots,p_n) \dfn  \negg \bigvee_{1\leq i \leq n} \dep{p_i}$. It was noted in  \cite{piil_complexity_2015} that $X$ satisfies $\max(p_1,\dots,p_n)$ if and only if for each assignment $s$ with domain $\{p_1,\dots,p_n\}$ there is an expansion $s'$ of $s$ in $X$.
Let $\phi'\dfn \max(p_1, \dots, p_{n},q^1_1,\dots,q^k_{k})\wedge \psi$ denote the $\PTL(\depop{})$-formula,
where $\psi$ is obtained by using the following recursive translation to eliminate every quantifier from $\phi$ starting from left to right.
Each quantifier of type $\exists q^i_j$ is recursively translated as  $\Big(\dep{\vec{c}^i_j,q^i_j} \lor (\dep{\vec{c}^i_j,q^i_j} \land \psi)\Big)$.  Each quantifier of type $\U q^i_j$ is recursively translated as $\negg\Big(\dep{\vec{c}^i_j,q^i_j} \lor (\dep{\vec{c}^i_j,q^i_j} \land \negg\psi)\Big)$. For the right most quantifier in the recursive translation, we set $\psi\dfn \theta$. It is quite straightforward to prove (cf.  \cite[Theorem 7]{piil_complexity_2015}) that $\phi$ is true under the constraint $(C^1_1,\dots,C^k_{j_k})$ if and only if $\phi'$ is satisfiable.
Here the connection between the classes $\SigmaE{k}$ and $\PiE{k}$, and the fragment of $\PTL(\depop{})$ of sentences with $\degg{\phi}\leq k$ is even more tighter than above. We obtain $\SigmaE{k}$-hardness for $\SAT(\LL)$, where $\LL$ is the fragment of $\PTL(\depop{})$ with formulae with $\degg{\phi}\leq k$. Note that using the above recursive translation and by setting $\max(p_1, \dots, p_{n},q^1_1,\dots,q^k_{k})\dfn \forall p_1 \cdots \forall p_{n}\forall q^1_1\dots\forall q^k_{k}$, we obtain $\SigmaE{k}$-hardness for $\TRUE(\LL)$, where $\LL$ is the fragment of $\QPTL(\depop{})$ with formulae with $\degg{\phi}\leq k$.

Finally note that dependence atoms of type $\dep{p_1,\dots,p_n,q}$ can be expressed via unary atoms as follows
\[
\negg \big( (r\lor \neg r) \lor \bigwedge_{1\leq n} \dep{p_i} \land \negg \dep{q} \big),
\]
while unary atoms $\dep{p}$ can be rewritten as  $\negg( \negg p \land \negg\neg p)$. As a summary, we obtain the following results.

\begin{proposition}\label{thm:adqbf-to-ptl2}
Every $\ADQBF$-instance $(\psi, \tuple C)$ can be translated in polynomial time to a $\QPTL$-sentence $\phi$ \suchthat $(\psi,\tuple C)$ is true iff $\phi$ is true.
\end{proposition}

\begin{proposition}\label{thm:ptl-hardness-levels}
For a logic $\LL$ let $\LL_k$ denote the fragment of $\LL$ with formulae $\phi$ for which $\degg{\phi}\leq k$. Then $\TRUE(\QPTL_k(\depop{}))$ and $\SAT(\PTL_k(\depop{}))$ are in $\SigmaE{k+1}$ and $\SigmaE{k}$-hard \wrt $\leqpm$-reductions. Moreover  $\TRUE(\QPTL_k)$ and $\SAT(\PTL_k)$ are in $\SigmaE{k+1}$ and $\SigmaE{k-2}$-hard \wrt $\leqpm$-reductions.
\end{proposition}

\section{Generalized dependence atoms}\label{sec:atoms}

In this section we study extensions of $\QPL$ and $\QPTL$ by the so-called generalized dependence atoms. In the context of first-order dependence logic, generalized atoms were introduced by Kuusisto \cite{Kuusisto2015}.

An \emph{$n$-ary generalized dependence atom} ($n$-GDA) is a set $G$ of $n$-ary relations over the Boolean domain $\{0,1\}$. For each $n$-GDA $G$, we introduce an atomic expression $A_G(p_1,\dots,p_n)$ that takes $n$ proposition symbols as parameters. Let $X$ be a team with $\{p_1, \ldots ,p_n\}\subseteq \Dom{X}$. The satisfaction relation $X\models  A_G(p_1,\dots,p_n)$ is given as follows:
\[
X\models A_G(p_1, \ldots ,p_n)  \Leftrightarrow  \rel{X,(p_1,\dots,p_n)}\in G,
\]
where $\rel{X,(p_1, \dots,p_n)} \dfn \{\big(s(p_1),\dots, s(p_n) \big) \mid s\in X  \}$.
We say that an $\SOQPL$-formula $\phi(f)$ with free function variable $f$ \emph{defines} $G$ if
\[
X\models A_G(p_1, \ldots ,p_n) \Leftrightarrow \emptyset^f_{\chi(X, (p_1, \dots,p_n))} \models \phi(f),
\]
where $\chi(X, (p_1, \dots,p_n))$ is the characteristic function of $\rel{X,(p_1, \dots,p_n)}$, and $\emptyset^f_{\chi(X, (p_1, \dots,p_n))}$ is the assignment that maps $f$ to $\chi(X, (p_1, \dots,p_n))$.
Moreover, we call $G$ $\SOQPL$-\emph{definable} ($\ESOQPL$-\emph{definable}) if there exists an $\SOQPL$-formula ($\ESOQPL$-formula) $\phi(f)$ that defines $G$. For a set $\calG$ of GDAs, let us denote by $\QPL(\calG)$ ($\QPTL(\calG)$) the logic obtained by extending $\QPL$  ($\QPTL$)  with the atoms in $\calG$.
For a set $\calG=\{G_i \ |\  i\in \mathbb{N}\}$ of atoms and respective defining sentences $\phi_i$, the set  $\calG$ is said to be \emph{polynomial time translatable} if the function $1^n \mapsto \langle \phi_n\rangle $, where  $\langle \phi_n\rangle$ is the binary encoding of $\phi_n$, is polynomial-time computable. The following theorem relates the logics  $\QPL(\calG)$ and $\QPTL(\calG)$ to   $\ESOQPL$ and $\SOQPL$, respectively.

As an example, we consider the dependence atom introduced in \Cref{sec:team}.

\begin{example}
The $n$-ary dependence atom $\dep{p_1,\ldots,p_{n-1},q}$ corresponds to the $n$-GDA that is defined as
\[
\Set{ R \subseteq \{0,1\}^n | \forall (s_1,\ldots,s_n) ,(t_1,\ldots,t_n) \in R : \bigwedge_{i=1}^{n-1} s_i = t_i \text{ implies } s_n = t_n }
\]

It is definable (even without second-order quantifiers) by the $\ESOQPL$-formula
\[
\phi(f) \dfn \forall x_1 \dots \forall x_n \forall y_1\dots \forall y_n \left(f(x_1, \ldots, x_n) \land f(y_1,\ldots,y_n) \land \bigwedge_{i=1}^{n-1} x_i \leftrightarrow y_i\right) \rightarrow (x_n \leftrightarrow y_n)
\]
\end{example}

\begin{theorem}\label{thm:QPL(G)}
Let $\calG$ be a set of $\ESOQPL$-definable \emph{(}$\SOQPL$-definable\emph{)}, polynomial time translatable  generalized dependence atoms. Then every sentence in $\QPL(\calG)$ \emph{(}$\QPTL(\calG)$\emph{)} can be translated to an equivalent $\ESOQPL$ \emph{(}$\SOQPL$\emph{)} sentence in polynomial time.
\end{theorem}

These translations are analogously presented by Väänänen in the first-order setting \cite{vaananen07}.
As we restrict ourselves to propositional logics, the difference is that only the domain $\{0,1\}$  and therefore the logic $\SOQPL$ are considered for the resulting formulae. The idea is to encode teams of assignments as their Boolean "characteristic functions". We start the proof with a slightly more general lemma.

\begin{lemma}\label{thm:gda-to-so-ptime}
Let $\calG$ be a set of $\SOQPL$-definable, polynomial time translatable generalized dependence atoms. Then for every formula $\phi\in\QPTL(\calG)$ and every set of proposition symbols $ \{p_1, \ldots ,p_n\}\supseteq  \Fr{\phi}$ there is an $\SOQPL$-sentence $\psi(f)$ computable in polynomial time \suchthat for all Boolean teams $X$ with  $\Dom{X} \supseteq \{p_1, \ldots, p_n\} $,
\[
X \models \phi\Leftrightarrow \emptyset^f_{\chi(X, (p_1, \dots,p_n))} \models \psi(f)\text{.}
\]
\end{lemma}
\begin{proof}

Analogously to the textbook translation from dependence logic to $\ESO$ \cite{vaananen07}, we show how to transform an open $\QPTL(\calG)$-formula $\phi$, whose free proposition symbols are from $ \{p_1, \ldots ,p_n\}$, to an $\SOQPL$-sentence $\phi^*(f)$. One can easily verify by induction that for any team $X$ with $\Dom{X}\supseteq\{p_1, \ldots ,p_n\} $, $X\models  \phi$ if and only if  $\emptyset^f_{\chi(X, (p_1, \dots,p_n))} \models \phi^*(f)$.
The construction proceeds recursively as follows.
\begin{enumerate}
\item Assume $\phi = A_G(p_{i_1}, \ldots ,p_{i_k})$ for $G\in \calG$, and let $\tuple p=(p_1, \ldots ,p_n)$, $\tuple p_0= (p_{i_1}, \ldots ,p_{i_k})$. Moreover, let $\tuple p_1$ be any sequence listing $\{p_1, \ldots ,p_n\}\setminus\{p_{i_1}, \ldots ,p_{i_k}\}$. Then $\phi^*(f)$ is defined as
\[
\exists g \, \psi(g)\wedge \pi_{\tuple p, \,\tuple p_0}(f,g),
\]
where $g$ has arity $k$, $\psi(g)$ is the $\SOQPL$-translation of $G$ and $\pi_{\tuple p, \,\tuple p_0}(f,g) \dfn \forall \tuple p (f(\tuple p) \rightarrow g(\tuple p_0))\wedge \forall \tuple p_0\exists \tuple p_1(g(\tuple p_0) \rightarrow f(\tuple p))$ expresses that the team encoded in $g$ is the projection of the team encoded in $f$ onto the variables $\tuple p_0$.
\item If $\phi = p_i$, then $\phi^*(f)$ is defined as
$\forall \tuple p (f(\tuple p) \rightarrow p_i).$

\item If $\phi = \neg p_i$, then
$\phi^*(f)$ is defined as
$\forall \tuple p (f(\tuple p) \rightarrow \neg p_i).$

\item If $\phi= \psi_0\wedge \psi_1$, then
$\phi^*(f)$ is defined as
$\psi^*_0(f)\wedge \psi^*_1(f).$

\item If $\phi= \psi_0\vee \psi_1$, then
$\phi^*(f)$ is defined as

\[
\exists f_0\exists f_1(\psi^*_0(f_0)\wedge\psi^*_1(f_1) \wedge \forall \tuple p\big (f(\tuple p) \rightarrow (f_0(\tuple p) \vee f_1(\tuple p)))\wedge \forall\tuple q ((f_0(\tuple q) \vee f_1(\tuple q))\rightarrow f(\tuple q) )\big ).
\]

\item If $\phi = \forall q \psi$, then $\phi^*(f)$ is defined as
\[
\exists g\big (\psi^*(g)\wedge \forall \tuple p\; \forall p'\,(f(\tuple p) \rightarrow g(\tuple p ,p')) \wedge \forall q \forall q' ( g(\tuple q, q') \rightarrow f(\tuple q))\big).
\]

 \item If $\phi = \exists q \psi$, then $\phi^*(f)$ is defined as
\[
\exists g\big (\psi^*(g)\wedge \forall \tuple p\;\exists p'  \big ( f(\tuple p) \imp g(\tuple p, p'))   \wedge \forall \tuple q \; \forall q' ( g(\tuple q, q')\imp f(\tuple q) \,)\big ).
\]
\item If $\phi=\negg \psi$, then $\phi^*(f)$ is defined as $\neg \psi^*(f)$.\qedhere
\end{enumerate}
\end{proof}

\medskip

Note that, if the atoms in $\calG$ are $\ESOQPL$-definable and if the $\negg$-case is dropped from the above translation, then the resulting formula itself is in $\ESOQPL$.

\begin{corollary}\label{thm:gda-to-eso-ptime}
Let $\calG$ be a set of  $\ESOQPL$-definable, polynomial time translatable generalized dependence atoms. Then for every formula $\phi\in\QPL(\calG)$ and every set of proposition symbols $ \{p_1, \ldots ,p_n\}\supseteq  \Fr{\phi}$ there is an $\ESOQPL$-sentence $\psi(f)$ computable in polynomial time \suchthat for all Boolean teams $X$ with  $\Dom{X} \supseteq \{p_1, \ldots, p_n\} $,
\[
X \models \phi\Leftrightarrow \emptyset^f_{\chi(X, (p_1, \dots,p_n))} \models \psi(f)\text{.}
\]
\end{corollary}

\noindent We are now ready to prove \Cref{thm:QPL(G)}.

\begin{proof}[Proof of \Cref{thm:QPL(G)}]
Let $\phi$ be a $\QPL(\calG)$ resp.\ $\QPTL(\calG)$ sentence. First translate it to an $\ESOQPL$-sentence resp.\ $\SOQPL$-sentence $\psi(f)$ in polynomial time according to \Cref{thm:gda-to-eso-ptime} resp.\ \Cref{thm:gda-to-so-ptime}.
It holds that $X \models \phi$ iff $\emptyset^f_{\chi(X,(p_1,\ldots,p_n))} \models \psi(f)$, if $\Dom{X} \supseteq \{p_1, \ldots, p_n\}$, and in particular that $\phi$ is true (\ie, satisfied by a non-empty team) iff $\psi' \dfn \exists f \exists \tuple p (f(\tuple p) \land \psi(f))$ is true.
\end{proof}

\noindent Hence, we may conclude this section with the following complexity results.

\begin{theorem}\label{thm:qptl-gda-aexppoly}

\begin{enumerate}[(i)]
\item Assume that $\calG$ is a polynomial time translatable set of $\SOQPL$-definable generalized dependence atoms. Then $\TRUE(\QPTL(\calG))$ is $\AEXPPOLY$-complete \wrt $\leqpm$-reductions.
\item Assume that $\calG$ is a polynomial time translatable set of $\ESOQPL$-definable generalized dependence atoms, and assume that dependence atoms translate into $\QPL(\calG)$ in polynomial time. Then $\TRUE(\QPL(\calG))$ is $\NEXPTIME$-complete \wrt $\leqpm$-reductions.
\end{enumerate}
\end{theorem}
\begin{proof}
Both upper bounds follow from \Cref{thm:QPL(G)} and \Cref{thm:qbsf-completeness}.  The lower bound for $\TRUE(\QPTL(\calG))$ ($\TRUE(\QPL(\calG))$) follows from \Cref{thm:adqbf-to-ptl2} (\Cref{TrueQPD}).
\end{proof}

\section{Summary}

In this article we compared different approaches to function quantification, with the logics depicted in \Cref{fig:summary}. We showed that, while some of the logics can express the quantification of functions only in a restricted way, like only in form of Skolem functions, they all can be efficiently translated into each other. It was shown in \Cref{thm:simple-prenex-so} that the "uniqueness" property of function symbols occurring in $\SOQPL$ and $\ESOQPL$ formulae can be obtained and hence (as depicted in the proof of \Cref{thm:odd-k-dqbf-hardness}) these formulae have a natural translation into $\ADQBF$ and $\DQBF$. \Cref{thm:adqbf-to-ptl} established that $\ADQBF$ and $\DQBF$ can easily be translated into team semantics, \ie, into $\QPTL(\depop{})$ and $\QPD$. The point is that the dependence atom can be used in team semantics to model the constraints of Skolem functions. Finally we showed in \Cref{thm:qptl-gda-aexppoly} 
that propositional team logic, even when augmented with generalized dependence atoms, can efficiently be translated back into $\SOQPL$ resp.\ $\ESOQPL$ when teams are modeled as Boolean functions.
Thus all these formalisms capture the same complexity classes: the class $\AEXPPOLY$ by unbounded quantifier alternation and the class $\NEXPTIME$ by the existential fragment.
Since $\QPTL(\calG)$ can express the dependence atom, it is complete for $\AEXPPOLY$ for any set $\calG$ of polynomial time translatable $\SOQPL$-definable generalized dependence atoms. For $\QPL(\calG)$ the matter is more complicated: If the dependence atom can be efficiently expressed in  $\QPL(\calG)$ and $\calG$ is a set of polynomial time translatable $\ESOQPL$-definable generalized dependence atoms, then $\QPL(\calG)$ is $\NEXPTIME$-complete, but for instance for $\QPLInc$ the complexity drops down to $\EXPTIME$, as shown in \Cref{thm:qplinc}. Higher levels of the exponential hierarchy are not only captured by fragments of $\SOQPL$ (see \Cref{thm:qbsf-completeness}), but also (with more or less sharp bounds) by the corresponding fragments of $\ADQBF$ (see \Cref{thm:odd-k-dqbf-hardness}) and $\mathsf{(Q)PTL}$ (\Cref{thm:ptl-hardness-levels}).

\medskip

\begin{figure}[ht]
\centering
    \begin{tabular}{llll}
    \toprule
	Logic&Ex. fragment&Method of function quantification&Example\\
	\midrule
    $\SOQPL$&$\ESOQPL$&Explicit, Second-order interpretation&$\exists f_y\, \forall \,\tuple x \, \phi$\\
    $\logicClFont{ADQBF}$&$\logicClFont{DQBF}$&Constraints, Skolem functions&$\forall \,\tuple x \, \exists y \, \phi$, $C_y = \{x_1\}$\\
    $\QPTL$&$\QPD$&Dep. atoms, Supplemented teams&$\forall \, \tuple x\, \exists y\; \phi \land \dep{x_1,y}$\\
    $\PTL$&$\PD$&Dep. atoms, Splitting of teams&$\dep{x_1,y} \lor (\phi \land \dep{x_1,y})$\\
    \bottomrule
\end{tabular}
\caption{Different formalisms of Boolean function quantification\label{fig:summary}}\end{figure}

\section*{Acknowledgements}
We wish to thank the anonymous referees for their helpful suggestions.
Miika Hannula was supported by the FRDF grant of the University of Auckland (project 3706751). Juha Kontinen and Jonni Virtema were supported by grant 292767 of the Academy of Finland.

\bibliographystyle{eptcs}
\bibliography{team_pl}

\end{document}